\documentclass[5p]{elsarticle} 
\pdfoutput=1 
\usepackage{graphicx,color,amssymb,amsmath,amsfonts,tikz,psfrag}
\usepackage{amsmath,amssymb,amsfonts,psfrag} 
\usepackage{graphicx}
\usepackage{caption}
\usepackage{subcaption}
\usepackage{subfig}
\usepackage{multi row}
\usepackage{array}
\usepackage{multicol}
\usepackage{sidecap}
\usepackage{tikz}
\usetikzlibrary{shapes,arrows}
\usepackage{enumerate}
\newtheorem{theorem}{Theorem}
\newtheorem{lemma}[theorem]{Lemma}
\newdefinition{remark}{Remark}
\newproof{proof}{Proof}
\newproof{pot}{Proof of Theorem \ref{thm2}}

\def\WVP{\mbox{\it{WVP}}}
\def\PWVP{\mbox{\it{PWVP}}}
\def\WV{\mbox{\it{WV}}}

\def\SPT{\mbox{\it{SPT}}}
\def\SP{\mbox{\it{SP}}}
\def\VP{\mbox{\it{VP}}}
\def\P{\mbox{\it{P}}}

\makeatletter
\def\ps@pprintTitle{%
  \let\@oddhead\@empty
  \let\@evenhead\@empty
  \def\@oddfoot{\reset@font\hfil\thepage\hfil}
  \let\@evenfoot\@oddfoot
}
\makeatother

\begin{document}
\title{Near Optimal Line Segment Weak Visibility Queries in Simple Polygons}

\author[add1]{Mojtaba Nouri Bygi}
\ead{nouribygi@ce.sharif.edu}
\author[add1,add2]{Mohammad Ghodsi}
\ead{ghodsi@sharif.edu}

\address[add1]{Computer Engineering Department, Sharif University of Technology, Iran}
\address[add2]{School of Computer Science, Institute for Research in Fundamental Sciences (IPM), Iran}

  \begin{abstract}
This paper considers the problem of computing 
the weak visibility polygon ($\WVP$) of any query line segment $pq$ (or $\WVP(pq)$) inside a given simple
polygon $\P$. We present an algorithm that preprocesses $\P$ and creates a data
structure from which  $\WVP(pq)$ is  efficiently reported in an output
sensitive manner. 

Our algorithm needs  $O(n^2 \log n)$ time and $O(n^2)$ space in the preprocessing phase to 
report  $\WVP(pq)$ of any query line segment $pq$ in time $O(\log^2 n + |\WVP(pq)|)$. 
We improve the preprocessing time and space of current results for this problem~\cite{nouri,chen2}
at the expense of more query time.
  \end{abstract}

  \maketitle

\section{Introduction \label{sec:intro}}

Two points inside a polygon $\P$ are {\em visible} to each other if their connecting segment remains
completely inside $\P$. The {\em visibility polygon} ($\VP$) of a point
$q$ inside $\P$ (or $\VP(q)$) is the set of vertices of $\P$ that are visible from $q$. There
have been many studies on computing $\VP$'s in simple polygons.
In a simple polygon $\P$ with $n$ vertices, $\VP(q)$ can be reported in time $O(\log n + |\VP(q)|)$ by
spending $O(n^3 \log n)$ time and $O(n^3)$ of preprocessing space \cite{bose}. 
This result was later improved by \cite{aronov} where the preprocessing 
time and space were reduced to $O(n^2 \log n)$ and $O(n^2)$ respectively, at the expense of more 
query time of $O(\log^2 n + |\VP(q)|)$.

The visibility problem has also been considered for line segments.
A point $v$ is said to be {\em weakly visible} from a line segment $pq$ if there exists
a point $w \in pq$  such that $w$ and $v$ are visible to each other.
The problem of computing  the {\em weak visibility polygon} of $pq$ (or $\WVP(pq)$) inside  $\P$ is 
to compute all points of $\P$ that are weakly visible from $pq$.
If $\P$ is simple (with no holes), Chazelle and Guibas \cite{Chazelle} gave an $O(n \log n)$ 
time algorithm for this problem. Guibas {\em et al.}\ \cite{guibas}
showed that this problem can be solved in $O(n)$ time if a triangulation of $\P$
is given along with $\P$. Since any $\P$ can be triangulated in $O(n)$ \cite{Chazelle2}, 
the algorithm  of Guibas {\em et al.}\ always runs in $O(n)$ time \cite{guibas}. Another 
linear time solution was obtained independently by \cite{toussainta}.

The $\WV$ problem in the query version has been considered by  few.
It was shown in \cite{bose} that a simple polygon $\P$ can be preprocessed 
in $O(n^3 \log n)$ time and $O(n^3)$ space
such that, given an arbitrary query line segment inside $\P$, 
$O(k \log n)$ time is required to recover $k$ weakly visible vertices.
This result was later improved by \cite{aronov} in which the preprocessing time and space were 
reduced to $O(n^2 \log n)$ and $O(n^2)$ respectively, at expense of more query 
time of $O(k \log^2 n )$.
In a recent work, we presented an algorithm to report $\WVP(pq)$
of any $pq$ in $O(\log n + |\WVP(pq)|)$ time by spending $O(n^3 \log n)$
time and  $O(n^3)$ space for preprocessing~\cite{nouri}.
Later, Chen and Wang considered the same problem and, by improving the 
preprocessing time of the visibility algorithm of Bose {\em et al.} \cite{bose},
they improved the preprocessing time to $O(n^3)$ \cite{chen2}.

In this paper, we show that the $\WVP$ of a line segment
$pq$ can be reported in near optimal time of $O(\log^2 n + |\WVP(pq)|)$, 
after preprocessing the input polygon in time and space of $O(n^2 \log n)$
and $O(n^2)$ respectively. 
Compared to the algorithms in \cite{nouri} and \cite{chen2}, the storage and preprocessing time has one 
fewer linear factor, at expense of more query time of $O(\log^2 n + |\WVP(pq)|)$. 
Our approach is inspired by Aronov {\em et al.} algorithm for
computing the visibility polygon of a point~\cite{aronov}. 
In Section \ref{sec:partial}, we first show how to compute the 
{\em partial weak visibility polygon} 
$\WVP(pq) \cap P'$ when $pq$ is not inside a sub-polygon $P'$ of $P$.
Then, in Section \ref{sec:balanced}, we use a balanced triangulation to compute and 
report the final weak visibility polygon.

\section{Preliminaries} \label{sec:pre}

In this section, we introduce some basic terminologies used throughout the paper.
For a better introduction to these terms, we refer the readers to 
Guibas {\em et al.}\ \cite{guibas}, Bose {\em et al.} \cite{bose}, and Aronov {\em et al.} \cite{aronov}.
For simplicity, we assume that no three vertices of the polygon are collinear.

\subsection{Visibility decomposition} \label{sec:pre:decompos}
Let $\P$ be a simple polygon with $n$ vertices. Also, let $p$ and  $q$ be 
two points inside $\P$.
The {\em visibility sequence} of a point $p$ is 
the sequence of vertices 
and edges of $\P$ that are visible from $p$. 
A {\em visibility decomposition} of $\P$ is to partition $\P$ into a set of 
{\em visibility regions}, such that any point inside each region has the same visibility sequence. 
This partition is induced by 
the {\em critical constraint edges},
which are the lines in the polygon each induced by two vertices of $\P$,
such that the visibility sequences of the points
on its two sides are different.

In a simple polygon, the visibility sequences of two  {\em neighboring} visibility regions which are separated by an edge, 
differ only in one vertex.
This fact is used to reduce the space complexity of maintaining the 
visibility sequences of the regions \cite{bose}.
This is done by defining the {\em sink regions}. 
A sink is a region with the smallest visibility sequence compared to all 
of its adjacent regions.
Therefore, it is sufficient to maintain the visibility sequences of the sinks,
from which the
visibility sequences of all other regions can be computed.
By constructing a directed dual graph over 
the visibility regions (see Figure~\ref{fig:f2}), one can maintain the difference between the visibility sequences of 
the neighboring regions \cite{bose}.

\begin{figure}[h]
  \centering
  \includegraphics[width=.8\columnwidth]{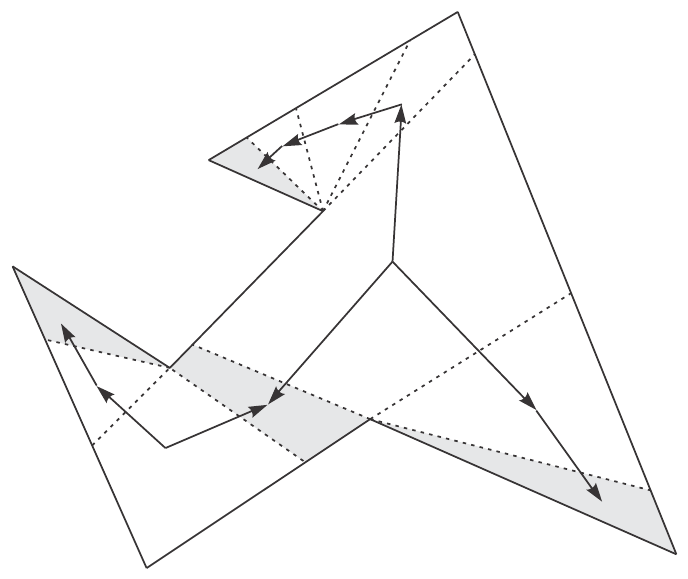}   
  \caption{The visibility decomposition induced by the critical constraint edges and its dual graph . The sink regions are shown in grey.}
  \label{fig:f2}
\end{figure}

In a simple polygon with $n$ vertices, the number of visibility
and sink regions are $O(n^3)$ and $O(n^2)$, respectively~\cite{bose}.

\subsection{A linear time algorithm for computing $\WVP$} \label{sec:guibas}
Here, we present the $O(n)$ time algorithm of Guibas {\em et al.}\ for computing
$\WVP(pq)$ of a line segment $pq$ inside $\P$, 
as described in \cite{ghosh}.
This algorithm is used in computing the partial weak visibility polygons in an output 
sensitive way, to be explained in Section~\ref{sec:partial:wvp}.
For simplicity, we assume that $pq$ is a convex edge of $P$,
but we will show that this can be extended for any line segment in the polygon.

Let $\SPT(p)$ denote the shortest path tree in $\P$ rooted at $p$. 
The algorithm traverses $\SPT(p)$ using a DFS and checks the turn at each vertex
$v_i$ in $\SPT(p)$. If the path makes a right turn at $v_i$, then 
we find the descendant of $v_i$ in the tree with the largest index $j$ (see Figure \ref{fig:f1}).
As there is no vertex between $v_j$ and $v_{j+1}$,
we can compute the intersection point $z$ of $v_jv_{j+1}$ and $v_kv_i$ 
in $O(1)$ time, where $v_k$ is the
parent of $v_i$ in $\SPT(p)$.
Finally the counter-clockwise boundary 
of $\P$ is removed from $v_i$ to $z$ by inserting the segment $v_iz$.

Let $P'$ denote the remaining portion of $\P$. We follow the same procedure for
$q$. This time, the algorithm checks every vertex to see whether the path 
makes its first left turn. If so, we will cut the polygon at that vertex in a similar way. 
After finishing the procedure, 
the remaining portion of $P'$ would be the $\WVP(pq)$.

\begin{figure}[h]
  \centering
  \includegraphics[width=1\columnwidth]{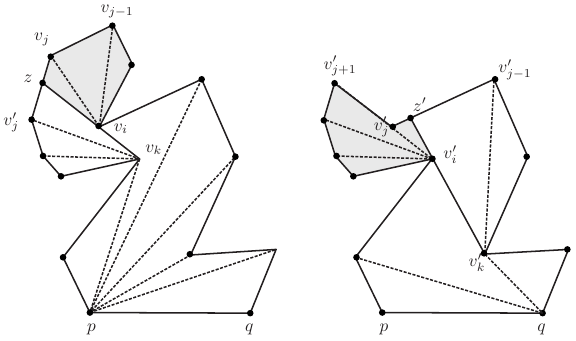}   
  \caption{ The two phases of the algorithm of computing $\WVP(pq)$.
	In the left figure, the shortest path from $p$ to $v_j$ makes a first right turn at $v_i$. In 
	the right figure, the shortest path from $q$ to $v'_j$ makes a first left turn at $v'_i$.}
  \label{fig:f1}
\end{figure}

\section{Computing the partial $\WVP$} \label{sec:partial}
Suppose that a simple polygon $\P$ is divided by a diagonal $e$ into two parts, $L$ and $R$.
For a query line segment $pq \in R$, we define the {\em partial weak visibility polygon}
$\WVP_L(pq)$ (or $\PWVP_L(pq)$ for clarity) to be the polygon $\WVP(pq) \cap L$.
In other words, $\WVP_L(pq)$ is the portion of $\P$ that is weakly visible from $pq$ {\em through} $e$.
In this section, we will show how to compute $\WVP_L(pq)$.
Later in Section \ref{sec:balanced}, we will use this structure to compute $\WVP(pq)$.

\begin{figure}[h]
  \centering
  \includegraphics[width=0.9\columnwidth]{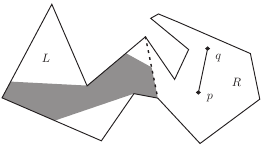}   
  \caption{The partial weak visibility polygon of the segment $pq$
  is defined as the part of the sub-polygon $L$ that is weakly visible
  from $pq$.}
  \label{fig:total}
\end{figure}

We will show how to use the algorithm of Guibas {\em et al.} \cite{guibas} 
to compute $\WVP_L(pq)$. To do so, we preprocess the polygon so that
we can answer the visibility query in an output sensitive way. 
The idea is to compute the visibility decomposition of the polygon and,
for each decomposition cell, compute the potential shortest path tree structures.
As the number of visibility regions is $O(n^3)$,  the preprocessing cost of 
our approach would be high.

To overcome, we only consider the critical constraint edges that cut $e$.
The number of such constraint edges is $O(n)$ and the complexity of the
decomposition is reduced to $O(n^2)$. This decomposition can be computed in $O(n^2)$
time. 
We call this decomposition the {\em partial visibility decomposition} of $\P$ with respect to $e$.
The remaining part of this section shows
how to modify the linear algorithm of Guibas {\em et al.} \cite{guibas} 
so that $\WVP_L(pq)$ can be computed in an output sensitive way. 
First, we show how to compute the
shortest path trees, and then present our algorithm for computing $\WVP_L(pq)$.

\subsection{Computing the partial $\SPT_L(p)$} \label{sec:pspt}
We define the {\em partial shortest path tree} $\SPT_L(p)$ to be the subset 
of $\SPT(p)$ that lead to a leaf node
in $L$. In other words, $\SPT_L(p)$ is the union of the shortest paths from $p$ 
to all the vertices of $L$.
In this section, we show how to preprocess the polygon
$P$, so that for any given point $p \in R$, any part of $\SPT_L(p)$
can be traversed in an output sensitive way. 
The shortest path tree $\SPT_L(p)$ is composed of two
kinds of edges: the {\em primary edges} that connect the root $p$ to its direct visible
vertices, and the {\em secondary edges} that connect two vertices of $\SPT_L(p)$
(see Figure \ref{fig:pspt}).
We also recognize two kinds of secondary edges: The 1st type of secondary edges 
(1st type for short) are those edges that are connected to a primary edge, 
and the 2nd type are the ones that connect other vertices of the polygon.
Notice that if a point $p$ crosses a critical constraint and that constraint
does not cut $e$, then the structure of $SPT_L(p)$ would not change.

\begin{figure}[h]
  \centering
  \includegraphics[width=1\columnwidth]{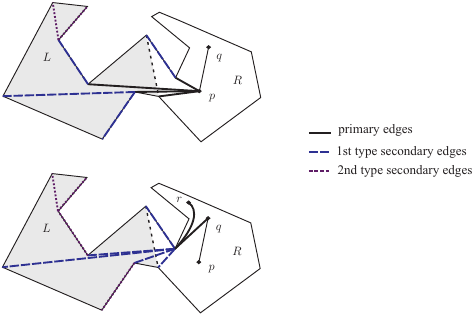}   
  \caption{$\SPT_L$ for different points of $R$. 
	Notice that as $q$ and $r$ are on the same visibility region w.r.t.\ $L$,   
   $\SPT_L(q)$ and $\SPT_L(r)$ have the same structure.}
  \label{fig:pspt}
\end{figure}

%
%
%

We can compute the primary edges of $\SPT_L$ by using Aronov's output sensitive 
algorithm of computing the partial visibility polygon \cite{aronov}. 
More precisely, with a 
processing cost of $O(n^2 \log n)$ time and $O(n^2)$ space,
giving a point $p$ in query time, a pointer to the sorted list of the 
vertices that are visible to $p$ can be computed in $O(\log n)$ time.

It is also necessary to compute the list of the secondary
edges of every vertex of $SPT_L$.
Each vertex $r$ in $SPT_L$ have $O(n)$ possible 2nd type edges. Depending on
the parent of $r$, a sub-list of these edges would appear in $\SPT_L$. 
To store all the possible 2nd type edges of $r$, 
we compute and store this sub-list, or to be precise, the starting and ending edges
of the list, for all the possible parents of $r$.
As there are $O(n)$
possible parents for a vertex, 
these calculations can be performed for all the vertices of the polygon in total time of
$O(n^2 \log n)$ and the data can be stored in $O(n^2)$ space. 
Having these data, we can, in the query time, access the list
of the 2nd type edges of any vertex in constant time.

We build the same structure for the 1st type edges. The parent of a 1st type
edge is the root of the tree. As the root can be in any of the $O(n^2)$ different
visibility regions, computing and storing the starting and ending edges in the
list of 1st type edges of a vertex cost $O(n^3 \log n)$ time and $O(n^3)$ space.

We can reduce the time and space needed to compute and store these structures,
having this property that two adjacent regions have only $O(1)$ differences in their 1st
type edges.

\begin{lemma} \label{lemm:lemma2}
Consider a visibility region $V$ in the polygon and
suppose that the 1st type secondary edges are computed for a point $p$
in this region.
For a neighboring region that share a common edge with $V$, 
these edges can be updated in constant time.
\end{lemma}
\begin{proof}
When a view point $p$ crosses the border of two neighboring regions, 
a vertex becomes visible or invisible to $p$ \cite{bose}.
In Figure \ref{fig:h4} for example, when $p$ crosses the
border specified by $u$ and $v$, a 1st type secondary edge of $u$
becomes a primary edge of $p$, and all the edges of $v$ become the 1st type secondary edges.
We can see that no other vertex is affected by this movement.
Processing these changes can be done in constant time, since it includes the following changes:
removing a secondary edge of $u$ ($uv$), adding a primary edge ($pv$), 
and moving an array pointer (edges of $v$) from the 2nd type edges of $uv$ to the 1st 
type edges of $pv$.
Note that we know the exact positions of these elements in their corresponding lists.
Finally, the only edge which involves in these changes
can be identified in the preprocessing time (the edge corresponding to the crossed critical constraint), 
so, the time we spent in the query time would be $O(1)$.
\end{proof}

\begin{figure}[h]
  \centering
  \includegraphics[width=.8\columnwidth]{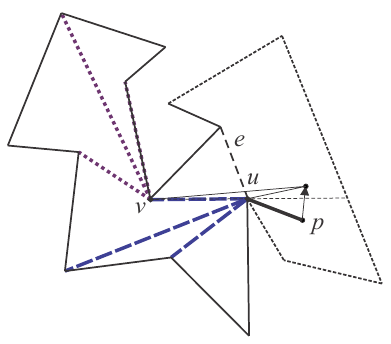}   
  \caption{When $p$ enters a new visibility region, 
  the combinatorial structure of $\SPT_L(p)$ can be maintained in constant time.}
  \label{fig:h4}
\end{figure}

Having this fact and using a {\em persistent data structure}, e.g. {\em persistent red-black tree} \cite{sarnak}, 
we can reduce the cost of
storing the 1st type edges by a linear factor.
A persistent red-black tree is a red-black tree 
that can remember all its intermediate versions.
If a set of $n$ linearly ordered items are stored it the tree and
we perform $m$ update into it, any version $t$, for $1 \leq t \leq m$, can be 
retrieved in time $O(\log n)$. 
This structure can be constructed in $O((m + n) \log n)$ 
time by using $O(m + n)$ space.

\begin{theorem}\label{theo:theo-pspt}
A simple polygon $P$ can be processed into a data 
structure with $O(n^2)$ space and in $O(n^2 \log n)$ time so that for any query point $p$, the
shortest path tree from $p$ can be reported in $O(\log n + k)$, where $k$ is the size of
the tree that is to be reported.
\end{theorem}

\begin{proof}
First, we use Aronov's algorithm for computing the partial visibility polygon of $p$.
For this, $O(n^2)$ space and $O(n^2 \log n)$ time is needed in the preprocessing phase.
For the secondary edges,  $O(n^2 \log n)$ time and $O(n^2)$
space is needed to compute and store these edges.
Also, a point location structure is built on top of the arrangement.

In the query time, the partial visibility region of $p$ can be located in $O(\log n)$ to
have the sorted list of the visible vertices from $p$.
As the visible vertices from $p$ correspond to the primary edges of $\SPT_L$,
we also have the primary edges of $\SPT_L(p)$.

For the 1st type edges, a tour is formed to visit all 
the cells of the partial visibility decomposition.
From Lemma \ref{lemm:lemma2}, we can start from an arbitrary cell, walk along the tour, 
and construct a persistent red-black tree on the 1st type edges 
of $\SPT_L$ of a point in each cell.
As there are $O(n^2)$ cells and, each cell has $O(n)$ 1st type edges, 
the structure takes $O(n^2)$ storage and can be built in $O(n^2 \log n)$ preprocessing time.
Having this structure, the 1st type edges of the cell containing
$p$ can be retrieved from the persistent data structure in $O(\log n)$ time.

Finally, at each node of the tree, we have the list of 2nd type edges from that node.
Therefore, the cost
of traversing $\SPT_L$ is the number of visited nodes of the tree, plus the
initial $O(\log n)$ time. In other words, the query time is $O(\log n + k)$, where 
$k$ is the number of the traversed edges of the $\SPT_L$.
\end{proof}

\subsection{Computing $\WVP_L(pq)$} \label{sec:partial:wvp}

Now that that we show how to compute $SPT_L(p)$ for any point $p \in R$ in the query time,
we can use the linear algorithm presented in Section \ref{sec:guibas} for computing $\WVP$
of a simple polygon and modify it to compute $\WVP_L(pq)$ 
in an output sensitive way. As we can see in Figure \ref{fig:polygonal-edge},
the algorithm can be extended to the cases that $pq$ is not a polygonal edge.
\begin{figure}[h]
	\centering
	\includegraphics[scale=1.6]{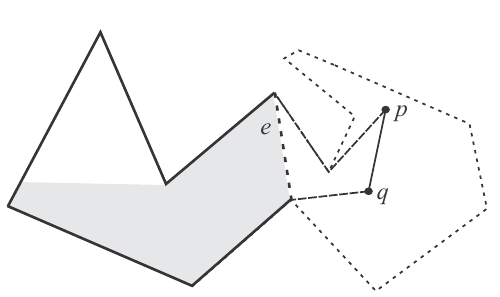}
    \caption{In computing $\WVP_L(pq)$ we can assume $pq$ to be a polygonal edge.}
	\label{fig:polygonal-edge}
\end{figure}

To achieve an output sensitive algorithm, we store some additional information 
about the vertices of the polygon in the preprocessing phase. 
We say that a vertex $v \in L$ is {\em left critical}
(LC for short) with respect to a point $q \in R$, if $\SP(q, v)$ makes its first left
turn at $v$ or one of its ancestors. In other words, each shortest path from $p$ to a
non-LC vertex is a convex chain that makes only clockwise turns at each node.
The {\em critical state} of a vertex is whether it is LC or not.
If we have the critical state of all the vertices of $L$ with respect to a point $q$, we say that we
have the {\em critical information} of $q$.

The idea is to change the algorithm of Section \ref{sec:guibas} and make it
output sensitive.
The outline of the algorithm is as follows:
In the first round, we traverse $SPT_L(p)$ using
DFS. At each vertex, we check whether this vertex is left critical with respect
to $q$ or not. If so, we are sure that the descendants of this vertex are not visible from
$pq$, so, we postpone its processing to the time it is reached from $q$, and check
other branches of $SPT_L(p)$. Otherwise, proceed with the algorithm and
check whether $SPT_L(p)$ makes a right turn at this vertex. In the second round,
we traverse $SPT_L(q)$ and perform the normal procedure of the algorithm. 

\begin{remark} \label{lemm:lemma3}
All the traversed vertices in $SPT_L(p)$ and $SPT_L(q)$ are 
vertices of $\WVP_L(pq)$.
\end{remark}

In the preprocessing phase, we compute the critical
information of a point inside each region, and assign this information to that region. In
the query time and upon receiving a line segment $pq$, we find the regions of $p$ and $q$. Using the
critical information of these two regions, the above algorithm can be applied to compute
$\WVP_L(pq)$.

As there are
$O(n^2)$ visibility regions in the partial visibility decomposition, $O(n^3)$ space 
is needed to store the critical information
of all the vertices. For each region, we compute $\SPT_L$ of a point, and by traversing
the tree, we update the critical information of each vertex with respect to this
region. An array of size $O(n)$ is assigned to each region to store these information.
We also build the structure described in Section \ref{sec:pspt} to compute $\SPT$ in
$O(n^3 \log n)$ time and $O(n^3)$ space. In the query time, we locate the visibility regions of
$p$ and $q$ in $O(\log n)$ time. By Remark \ref{lemm:lemma3}, when we proceed the algorithm 
in $\SPT_L$s of $p$ and $q$, 
we only traverse the vertices of $\WVP_L(pq)$.
Finally, as the processing time spent in each vertex
is $O(1)$, the total query time is $O(\log n + |\WVP_L(pq)|)$.

To improve this result, we use the fact that any two adjacent regions have
$O(1)$ differences in their critical information.

\begin{lemma} \label{lemm:lemma5}
In the path between neighboring visibility regions,
the changes of the critical information can be handled in constant time.
\end{lemma}

\begin{proof}
\begin{figure}	
	\centering
	\begin{subfigure}[b]{0.4\textwidth}
		\centering
		\includegraphics[width=.8\columnwidth]{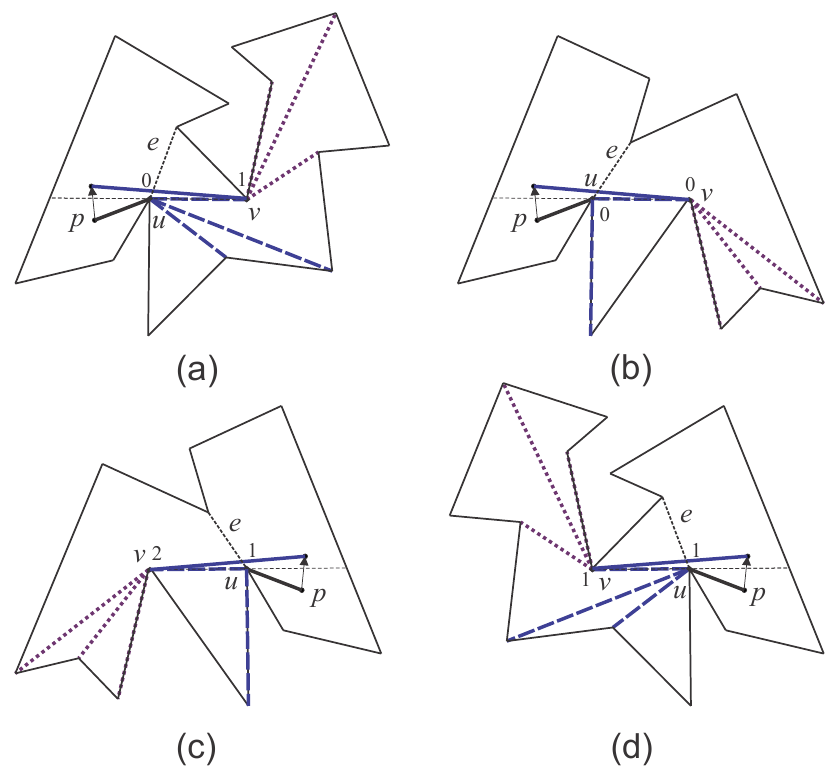}
		\caption{$v$ is LC but $u$ is not.}\label{fig:e2-1a}		
	\end{subfigure}
	\begin{subfigure}[b]{0.4\textwidth}
		\centering
		\includegraphics[width=.8\columnwidth]{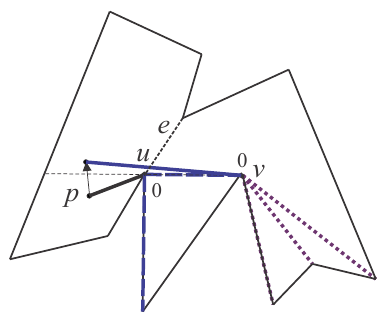}
		\caption{$u$ and $v$ are not LC.}\label{fig:e2-1b}		
	\end{subfigure}
	\\
	\begin{subfigure}[b]{0.4\textwidth}
		\centering
		\includegraphics[width=.8\columnwidth]{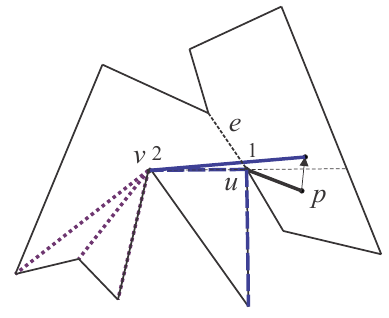}
		\caption{both $u$ and $v$ are LC.}\label{fig:e2-1c}		
	\end{subfigure}
	\begin{subfigure}[b]{0.4\textwidth}
		\centering
		\includegraphics[width=.8\columnwidth]{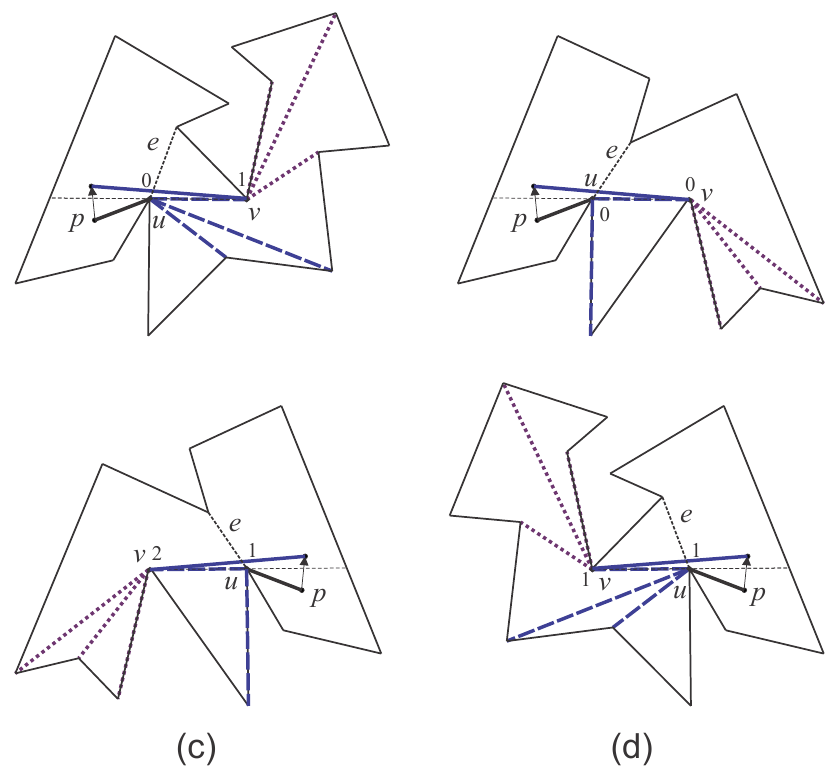}
		\caption{$u$ is LC but $v$ is not. }\label{fig:e2-1d}		
	\end{subfigure}
  \caption{Changes in the critical state of $v$ w.r.t.\ $p$, as $p$ moves between the two regions.}
  \label{fig:e2-1}
\end{figure}
Suppose that we want to maintain the critical information of $p$ and 
$p$ is crossing the critical constraint defined by $uv$, where $u$ and
$v$ are two reflex vertices of $P$. 
The only vertices that affect directly by this change are $u$ and $v$.
Depending on the 
critical states of $u$ and $v$ w.r.t.\ $p$, four situations may 
occur (see Figure \ref{fig:e2-1}). In the first three cases, the critical state of 
$v$ will not change. In the forth case, however,
the critical state of $v$ will change. Before the cross, the shortest path
makes a left turn at $u$, therefore, both $u$ and $v$ are LC w.r.t.\ $p$. However, after the cross, $v$
is no longer LC. This means that the critical state of all the children of $v$ in the $\SPT_L(p)$
may change as well.

To handle these cases, we modify the way the critical information of each vertex w.r.t.\ $p$
are stored. 
At each vertex $v$, we store two additional values: the number of LC vertices we met
in the path $SP(p,v)$ from $p$, or its {\em critical number}, and a {\em debit number}, 
which is the critical number that is to be propagated in the vertex subtree. 
It is clear that if a vertex is LC, it means that its critical number is greater 
than zero (see Figure \ref{fig:crtitical-numbers}). Also, if a vertex has 
a debit number, the critical numbers of all its children must be added by this
debit number. Notice that computing and storing 
these numbers along the critical information will not change the
time and space requirements.

\begin{figure}[h]
  \centering
  \includegraphics[width=.9\columnwidth]{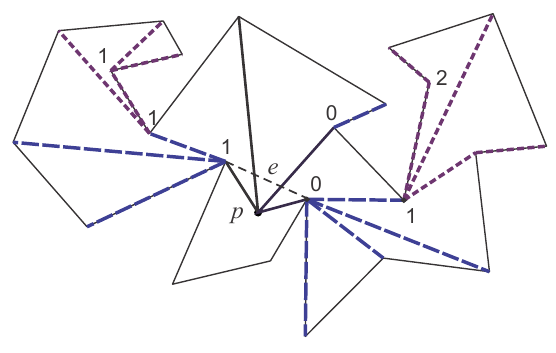}  
  \caption{
The critical number represents the number of the left critical vertices met from $p$ in $\SPT_e(p)$.}
  \label{fig:crtitical-numbers}
\end{figure}

Now let us consider the forth case in Figure \ref{fig:e2-1}. When $v$ becomes visible to $p$,
it is no longer LC w.r.t.\ $p$. Therefore, the critical number of $v$ is changed to $0$.
However, instead of changing the critical numbers of all the children of $v$, 
we set the debit number of $v$ to -1, indicating that the critical numbers of all the vertices of its subtree 
must be subtracted by 1. The actual propagation of this subtraction will happen 
at query time when $\SPT_L(p)$ is traversed. 
If $p$ moves in the reverse path, i.e., when $v$ becomes invisible to $p$, 
we handle the tree in the same way by storing $1$ in the debit numbers, and
propagating this addition in the query time.
\end{proof}

A persistent data structure can be used to reduce the costs to $O(n^2 \log n)$
preprocessing time and $O(n^2)$ storage. 
We form a tour visiting all the cells 
and construct a persistent red-black tree on the critical information and the 2nd
type edges of all the nodes. The structure takes $O(n^2)$ storage and can be built
in $O(n^2 \log n)$ preprocessing time. In addition, we build a point location
structure on top of the arrangement which can be done in $O(n^2)$ time and $O(n^2)$
space \cite{kirkpatrick}.

\begin{theorem}
\label{theorem:partial-visibility}
Given a polygon $P$ and a diagonal $e$ which cuts $P$ into two parts,
$L$ and $R$, and using $O(n^2 \log n)$ time, we can construct a data structure of size
$O(n^2)$ so that, for any query line segment $pq \in R$, the partial weak visibility
polygon $\WVP_L(pq)$ can be reported in $O(\log n + |\WVP_L(pq)|)$ time.
\end{theorem}

\section{Computing $\WVP$ by balanced triangulation}\label{sec:balanced}
There is always a diagonal 
$e$ of a simple polygon that cuts
$\P$ into two pieces, each having at most $2n/3$ vertices \cite{chazelle82}.
We can recursively subdivide and build a balanced binary tree
where the leaves are triangles and each interior node $i$ corresponds
to a subpolygon $P_i$ and a diagonal $e_i$. Each diagonal $e_i$ divides $P_i$
into two subpolygons, $L_i$ and $R_i$, which respectively correspond to the left and right subtrees of
$i$ (see Figure \ref{fig:triangulation}). 
We build the data structures described in Section \ref{sec:partial} for $L_i$ and $R_i$
with respect to $e_i$.

\begin{figure}[h]
	\centering
	\includegraphics[scale=1.6]{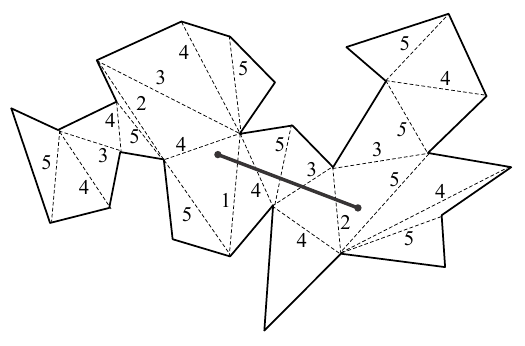}
	\caption{A balanced binary triangulation of the polygon is built so that 
	the the weak visibility polygon can be computed recursively.}
	\label{fig:triangulation}
\end{figure}

\begin{figure}[h]
	\centering
\begin{tikzpicture}
[
cross/.style ={draw=red, thick, circle,fill=red!20,minimum height=1em},
leaf/.style ={draw=blue, thick, circle,fill=blue!20,minimum height=1em},
normal/.style ={draw=yellow!4, thick, circle,fill=yellow!10,minimum height=1em},
level distance=8mm,
every node/.style={circle,inner sep=1pt},
level 1/.style={sibling distance=32mm},
level 2/.style={sibling distance=16mm},
level 3/.style={sibling distance=8mm},
level 4/.style={sibling distance=4mm},
level 5/.style={sibling distance=2mm}
]
\node [cross] {1}
	child {node [cross] {2}
		child {node {3}
			child {node {4}
				child {node {5}}
				child {node {5}}
			}
			child {node {4}
				child {node {5}}
				child {node {5}}
			}
		}
		child {node [cross] {3}
			child {node {4}
				child {node {5}}
				child {node {5}}
			}
			child {node [cross] {4}
				child {node {5}}
				child {node [leaf] {5}}
			}
		}
	}
	child {node [cross] {2}
		child {node [cross] {3}
			child {node [cross]{4}
				child {node [leaf] {5}}
				child {node [leaf] {5}}
			}
			child {node [leaf]{4}
				child {node {5}}
				child {node {5}}
			}
		}
		child {node [cross] {3}
			child {node [cross]{4}
				child {node [leaf]{5}}
				child {node {5}}
			}
			child {node {4}
				child {node {5}}
				child {node {5}}
			}
		}
	};
\end{tikzpicture}
\caption{The specified nodes correspond to the computed partial $\WVP$s 
in Figure \ref{fig:triangulation}}
\label{fig:tree}
\end{figure}
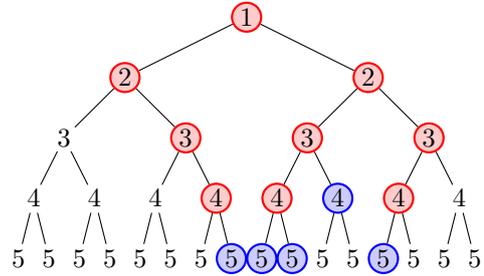

To compute $\WVP(pq)$, $p$ and $q$ will be located among the leaf triangles. In the simplest
case, both $p$ and $q$ belong to the same triangle. First we explain this situation.
We construct $\PWVP_i(pq)$ for
each $i$ from the leaf to the root. 
Here, $\PWVP_i(pq)$ is the partial weak visibility polygon of
$pq$ in $P_i$ with respect to $e_i$.
For the leaf node, it is the corresponding triangle, and
for other nodes, it can be computed inductively. In this case, the inductive step
is similar to that of \cite{aronov}. In each step, the merging of the computed polygons
can be done in $O(\log n)$.

The time and space needed for building an exterior visibility decomposition of a 
simple polygon with $m$ vertices are $O(m^2)$ and $O(m^2 \log m)$, respectively.
Thus, the space and time of the above inductive procedure 
can be expressed as the following equations:
\begin{align*}
S(n) &= \max_{n/3 \leq m \leq 2n/3} (S(m) + S(n-m)) + \Theta(n^2),\\
T(n) &= \max_{n/3 \leq m \leq 2n/3} (S(m) + S(n-m)) + \Theta(n^2 \log n)
\end{align*}
Therefore, $S(n) = \Theta(n^2)$, and $T(n) = \Theta(n^2 \log n)$.
With the same analysis as in \cite{aronov}, we can calculate the query time.
Two point locations can be done in $O(\log n)$ time. As the triangulation
is balanced, we path from the root to any node has $O(\log n)$ length.
As we showed in Theorem \ref{theorem:partial-visibility}, the time needed
to query $\PWVP_i(pq)$ at step $i$ is $O(\log n + |PWVP_i(pq)|)$.
Also, the merging at each step can be done in $O(\log n)$ time.
Therefore, the total query time is $O(\log n + \sum_i(\log n + |PWVP_i(pq|))$,
or $O(\log^2 n + |WVP(pq)|)$.

The tricky part is when $p$ and $q$ are on different triangles. Assume that at
step $i$, the query line segment is $p_iq_i$ and it is in sub-polygon $P_i$. 
The sub-polygon $P_i$ is divided by
diagonal $e_i$ to two sub-polygons $L_i$ and $R_i$. If $p_iq_i$ does not intersect $e_i$, 
without loss of generality, assume that $p_iq_i$ is located in $R_i$ (see Figure \ref{fig:inductive-step-1}).
We do the normal procedure of the algorithm and compute 
$\WVP_{L_i}(p_iq_i)$.
We continue to recursively compute weak visibility polygon on $R_i$.
In this case, the time needed by this step can be expressed as
$T(n_i, p_iq_i) = T(n_i/2, p_iq_i)+ O(\log n_i) + |\PWVP_{L_i}(p_iq_i)|$.

On the other hand, if $p_iq_i$ and $e_i$ intersect at point $r_i$,
without loss of generality, assume that $p_i$ is in $R_i$ and $q_i$ is in $L_i$
(see Figure \ref{fig:inductive-step-2}).
We can express $\WVP(p_iq_i)$ as the union of four weak visibility polygons:
\begin{enumerate}[i]
\item the partial weak visibility polygon of $p_ir_i$ on $L_i$, 
\item the weak visibility polygon of $p_ir_i$ in $R_i$, 
\item the partial weak visibility polygon of $r_iq_i$ on $R_i$, 
\item the weak visibility polygon of $r_iq_i$ in $L_i$. 
\end{enumerate}
In other words,
we must compute two partial weak visibility polygons, and two weak visibility sub-problems. 
Having these four visibility polygons, the union of them can be merged 
in time $O(4 |\WVP_{P_i}(p_iq_i)|)$.
According to the Theorem \ref{theorem:partial-visibility},
the query time spent at step $i$ can be expressed as: 
$T(n_i,p_iq_i) = O(\log n_i) + T(n_i/2,p_ir_i)+ T(n_i/2,r_iq_i)+ |\PWVP_{L_i}(p_ir_i)| + |\PWVP_{R_i}(r_iq_i)|$.

\begin{figure}	
	\centering
	\begin{subfigure}[b]{0.48\textwidth}
		\centering
		\includegraphics[scale=1.7]{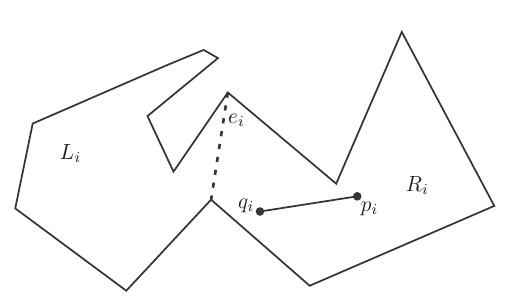}
		\caption{This case can be phrased as
		$\WVP(p_iq_i) = \PWVP_{L_i}(p_ir_i) + \WVP_{R_i}(p_ir_i)$}\label{fig:inductive-step-1}		
	\end{subfigure}
	\begin{subfigure}[b]{0.48\textwidth}
		\centering
		\includegraphics[scale=1]{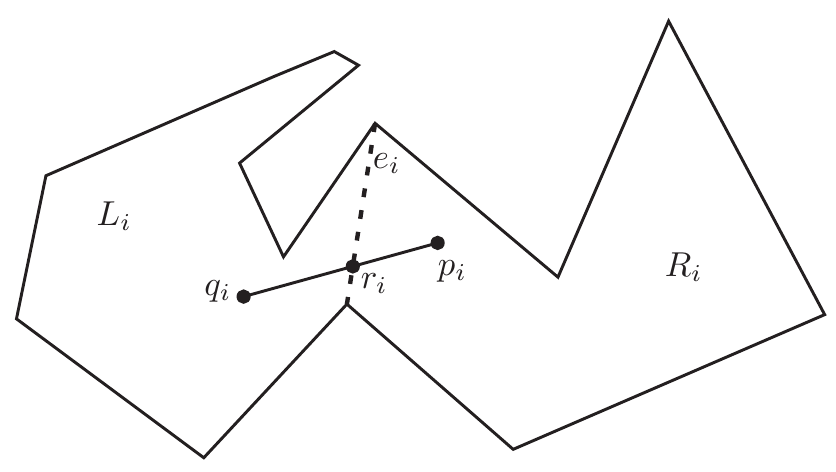}
 	\caption{This can be phrased as $\WVP(p_iq_i) = \PWVP_{L_i}(p_ir_i) 
 	+ \WVP_{R_i}(p_ir_i) + \PWVP_{R_i}(r_iq_i) + \WVP_{L_i}(r_iq_i)$.}\label{fig:inductive-step-2}		
	\end{subfigure}
\caption{The induction step can be categorized as one of these situations.}
  \label{fig:inductive-step}
\end{figure}

The preprocessing costs of the algorithm is the same as before. For the query
time, by induction, we can prove that $T(n,pq) = O(\log^2 n + |WVP(pq)|)$.
Assume that this property holds for every simple polygon with less than $n$ vertices and
every line segment $pq$ in it. For a simple polygon $P$ with $n$ vertices,
depending on whether $pq$ cuts the diagonal $e$, 
we have one of these equation:
\begin{align*}
  T(n,pq) &= T(\frac{n}{2}, pq) + O(\log n) + |\PWVP_e(pq)| \\
  		  \text{or,} \\
  T(n,pq) &= T(\frac{n}{2},pr) + T(\frac{n}{2},rq)+ O(\log n) \\
		   &   + |\PWVP_{e}(pr)| + |\PWVP_{e}(rq)|
\end{align*}
Here, $\PWVP_e(pq)$ is the partial weak visibility polygon of $pq$ w.r.t.\ the cut $e$.
In the first case, we have
\begin{align*}
  T(n,pq) &= O(\log^2 \frac{n}{2} + \log n + |\WVP_{\frac{n}{2}}(pq)| ) \\
	&= O(\log^2 n + |\WVP(pq)|)
\end{align*}
In the second case, we have 
\begin{align*}
  T(n,pq) &= O(\log n + \log^2 \frac{n}{2} + |\WVP_{\frac{n}{2}}(pr)| \\ 
  		 &+ |\WVP_{\frac{n}{2}}(rq)| +  |\PWVP_{e}(pr)| + |\PWVP_{e}(rq)|) \\
  		 &=  O(\log^2 n + |\WVP(pq)|)
\end{align*}
In these equations, $\WVP_{\frac{n}{2}}(pq)$ is the weak visibility polygon of 
$pq$ in a polygon of size $\frac{n}{2}$.

As for the base case, we showed that if $pq$ is located in single triangle, 
the time for computing $\WVP(pq)$ would be $O(\log^2 n + |\WVP(pq)|)$.
In summary, we have the following theorem:

\begin{theorem} \label{theoremxx}
A simple polygon $\P$ of size $n$ can be processed in $O(n^2 \log n)$ time into a
data structure of size $O(n^2)$ so that, for any query line segment $pq$, $\WVP(pq)$
can be reported in time $O(\log^2n + |\WVP(pq)|)$.
\end{theorem}

\section{Conclusion}

In this paper, we showed how to answer the weak visibility queries in a simple polygon
with $n$ vertices
in an efficient way. 
In the first part of the paper, 
we defined the partial weak visibility polygon $\WVP_e(pq)$
of a line segment $pq$ with respect to a diagonal $e$ and 
presented an algorithm to 
report it in time $O(\log n + |\WVP_e(pq)|)$, 
by spending $O(n^2 \log n)$ time to preprocess the polygon 
and maintaining a data structure of size $O(n^2)$. 

In the second part,
we presented a data structure of size $O(n^2)$ which can be computed in
time $O(n^2 \log n)$ so that the weak visibility polygon $\WVP(pq)$ from any query line segment $pq \in P$
can be reported in time $O(\log^2n  + |\WVP(pq)|)$.

\end{document}